\documentclass[11pt]{article}

\usepackage{amsfonts,amssymb,amsmath,latexsym,amsthm}
\usepackage{soul}
\usepackage{multirow,ae,aecompl,url} 
\usepackage[inline]{enumitem}

\usepackage{graphicx}
\usepackage{xcolor}
\usepackage{float}
\usepackage{braket}
\usepackage[ruled,vlined,linesnumbered]{algorithm2e}
\usepackage{hyperref}

\newcommand{\FFF}{\vspace*{\bigskipamount}}

\newcommand{\cC}{\mathcal{C}}

\newcommand{\cE}{\mathcal{E}}

\newcommand{\cO}{O}
\newcommand{\cP}{\mathcal{P}}

\newcommand{\Fuzzy}{Fuzzy }

\newcommand{\remove}[1]{}

\newcommand{\dk}[1]{{\color{red}{#1}}}
\newcommand{\jo}[1]{{\color{blue}{#1}}}

\newcommand{\Paragraph}[1]{\paragraph{#1}}

\newcommand{\polylog}{\text{polylog }}

\makeatletter
\newtheorem*{rep@theorem}{\rep@title}
\newcommand{\newreptheorem}[2]{%
\newenvironment{rep#1}[1]{%
 \def\rep@title{#2 \ref{##1}}%
 \begin{rep@theorem}}%
 {\end{rep@theorem}}}
\makeatother

\newtheorem{theorem}{Theorem}
\newreptheorem{theorem}{Theorem}
\newtheorem{lemma}{Lemma}

\newtheorem{definition}{Definition}

\newcommand{\ceil}[1]{\left \lceil #1 \right \rceil }
\newcommand{\floor}[1]{\left \lfloor #1 \right \rfloor }
\newcommand{\E}{\mathbb{E}}

\begin{document}
\title{
Fault-Tolerant Consensus in Quantum Networks \vfill
}

\author{  
MohammadTaghi Hajiaghayi \footnotemark[1]
\and Dariusz R. Kowalski \footnotemark[2]
\and Jan Olkowski \footnotemark[1]}

\footnotetext[1]{	Department of Computer Science, University of Maryland, College Park, Maryland, USA. Partially supported by DARPA QuICC, the NSF grant 2218678, and the NSF grant 2114269.}

\footnotetext[2]{	School of Computer and Cyber Sciences, Augusta University, Augusta, Georgia, USA.}




\date{}

\maketitle

\vfill

\begin{abstract}

Fault-tolerant \textit{consensus} is about reaching agreement on some of the input values in  a limited time by 
non-faulty
autonomous processes, despite of failures of processes or communication medium.
This problem is particularly challenging and costly against an adaptive adversary with full information. Bar-Joseph and Ben-Or~\cite{Bar-JosephB98} (PODC'98) were the first who proved an absolute lower bound $\Omega(\sqrt{n/\log n})$  on expected time complexity of consensus in any classic (i.e., randomized or deterministic) message-passing network with $n$ processes succeeding with probability $1$ against such a strong adaptive adversary crashing processes.

Seminal work of Ben-Or and Hassidim~\cite{Ben-OrH05} (STOC'05) broke the $\Omega(\sqrt{n/\log n})$ barrier for consensus in classic (deterministic and randomized) networks 
by employing quantum computing.
They showed an (expected) constant-time quantum algorithm for a linear number of crashes $t<n/3$.

In this paper, we improve upon that
seminal work
by reducing the number of quantum and communication bits
to an arbitrarily small polynomial, and even more, to a polylogarithmic number -- though, the latter in the cost of a slightly larger polylogarithmic time (still exponentially smaller than the time lower bound $\Omega(\sqrt{n/\log n})$ for classic computation).

\vfill 

\noindent
\textbf{Keywords:} 
distributed algorithms, 
quantum algorithms, 
adaptive adversary, 
crash failures, 
Consensus, 
quantum common coin, 
approximate counting
\end{abstract}

\vfill

\thispagestyle{empty}

\setcounter{page}{0} 

\newpage



\section{Introduction}

\remove{Reviewer 1: 
"Second, there are numerous typos and language errors that make reading very hard."

"Also it should be mentioned that the first part of the introduction is almost a copy of the abstract, which is a waste of the reader's time." PROBABLY WE SHOULD REPHRASE IT?
}

\remove{Reviewer 2: The main improvement of this work comes from the graph / gossipping technique, but not the quantum power it self. I think combining the existing tools to achieve the communication-efficient quantum coin is interesting. However, I might not consider it as solving a fundamental difficulty in consensus or quantum computing. (Just to remember that we might want to give better explenation why the problem is important.)
}

{\em Consensus} is about making a common decision among the processes' initial input values in  a limited time by every 
non-faulty
process, despite the faulty behaviour of some of the players. 
Since its introduction by Pease, Shostak and Lamport~\cite{PSL} (JACM'80), who ruled out trivial solutions (such as always deciding on the same bit), fault-tolerant consensus has constantly been among foundation problems in distributed computing. 
This problem has been studied in synchronous/asynchronous and deterministic/randomized computation models and under various fault-tolerant or adversarial models: Fail-Stop (crashes) and Byzantine, Static and Adaptive, Computationally Bounded and Unbounded Adversaries - just to name a few (see Section~\ref{sec:related} for related work). 


While the landscape of Consensus problem under the classic model of computation is well-developed, much less is known if we allow for quantum computation. The seminal work of Ben-Or and Hassidim~\cite{Ben-OrH05} (albeit a short 5-pager in STOC'05) broke the $\Omega(\sqrt{n/\log n})$ rounds time barrier for classic computing by employing quantum computing to further use the power of randomization in distributed computing.
They showed an (expected) constant-time quantum algorithm for a linear number of crashes $t<n/3$, however, the algorithm is inefficient in terms of communication bits and, even more importantly, in terms of the number of quantum bits (qubits), as it uses $\Omega(n)$ of them per process. Since then no algorithm has managed to reduce the quantum resources needed to solve Consensus. Because generating, maintaining and sending quantum bits is extremely costly (today's quantum devices use less than 500 qubits), thus the main question of our paper emerges~naturally:

{\em Could the number of qubits be substantially reduced without harming the time complexity?}


\Paragraph{Distributed setting.} 
\remove{Reviewer 2: I think this work consider the full information adversary model [7], while it is not mentioned throughout the paper. If we are not in the full information model, then the results are incomparable to [7].}

We consider a quantum synchronous message-passing model (c.f.,~\cite{AW}), consisting of $n$~synchronous processes (also called players), each with common clock (a clock tick is called a round or a step) and unique id from the known set $\cP=[n]=\{1,\ldots,n\}$.  

Between any pair of processes we assume the existence of a quantum channel being able to transmit reliable\footnote{%
Messages are not lost nor corrupted while in transit.}
messages caring quantum bits, qubits. For the sake of completeness,  we also augment the model by classic point-to-point channels between any pair of processes.
In each round, a process can send (personalized) quantum and classic messages to any selected subset of other processes.
After {\em multicasting messages}, in the same round
a process {\em receives messages} that were just sent to it by other processes, and performs {\em local computation}, which involves both quantum and classic bits.\footnote{%
Local computation also decides what messages to send in the next round and to whom.}
%

Processes are prone to {\em crash failures}, also called {\em fail-stops}. 
A crashed process permanently stops any activity, including sending and receiving messages.

We model crashes as incurred by a full-information {\em adversary} (the same as in~\cite{Bar-JosephB98,Ben-OrH05}) that knows the algorithm, the exact pure quantum state (see Section~\ref{sec:quantum-prel}) and the classic state of the system at any point of an execution, and has an unbounded computational power. The adversary decides which processes to fail and when. The adversary 
is also {\em adaptive} -- 
it can make a decision on-line based on its current full-knowledge of the system. 
However, the adversary does not know the future computation, which means that it does not know future random bits drawn by processes.

As for the quantum part, the adversary can apply no quantum operation to the system, but it is aware of all quantum and classic changes of state that the network undergoes. If a process is crashed by the adversary, we assume that its quantum bits are not destroyed (in particular, entangled qubits in other processes do not collapse but maintain their entanglement), however they cannot be used in 
further~computation.

Failures are {\em not clean} -- when a process crashes when attempting to multicast a message, then some of the recipients may receive the message and some may not; this aspect is controlled by the adversary. 
A {\em $t$-adversary} is additionally restricted by the the number of crashed processes being {\em smaller than~$t$}; if $t=n$ then the $n$-adversary is also called an unbounded adversary (note that even in such case, at least one process must survive for consensus to make sense).
Throughout the paper, we will be calling the adversary described above ``adaptive'',~for~short.

\Paragraph{{\em Consensus} problem:} each process $p$ has its own initial 
value $input_p$ and has to output a (common) decision value, so that the following conditions hold:
{\em validity} -- decision must be on an input value of some process; {\em agreement} --	no two processes decide on different values; and {\em termination} -- each process eventually decides on some value, unless it is faulty.
All those three requirements must hold with probability 1. We focus on {\em binary consensus}, in which initial values are in $\{0,1\}$. 

Correctness and complexity -- in terms of time (the number of rounds needed for all processes to terminate) and the 
number of quantum bits (qubits) and communication bits -- are analyzed and maximized (worst-case) against an adaptive adversary.

We say that a random event occurs \emph{with high probability}
({\em whp} for short), if its probability could be made $1-\cO(n^{-c})$ for any sufficiently large
positive constant $c$ by linear scaling of parameters.

\remove{
First, is the Hadamard gate $H^{k}$. When applied on a quantum state $\ket{00\ldots 0}$ of $k$ qubits, it creates an equal superposition of all the possible
states of $k$ qubits. Formally, this works as follows: $H^{k}\ket{00\ldots 0} = \sum_{x} \frac{1}{\sqrt{2^{k}}} \ket{x}$, where $x$ iterates over all possible binary strings of length $k$. Second is the $CNOT$ gate applied on two qubits. If $\ket{x} = \alpha\ket{0} + \beta\ket{1}$ then $CNOT(\ket{x}, \ket{y}) = \alpha\ket{0}\ket{y} + \beta\ket{1}\ket{y}$. Informally, $\ket{x}$ is the control qubit, if it is $\ket{1}$ then the bit-wise operation applies to the qubit $\ket{y}$, otherwise identity is applied to $\ket{y}$. In general, any quantum operation $U$ we know how to implement can be extended to one controlled by an additional qubit. If $\ket{x}$ is the control qubit then the
controlled operation acts like $U$ when $\ket{x} = \ket{1}$ and it acts like identity when $\ket{x} = \ket{0}$.}
\section{Our Results}
In this work, we focus on improving quantum bits' and communication complexities (without harming time complexity) of quantum algorithms  solving Consensus problem with probability 1 against an adaptive full-information adversary capable of causing 
processes' crashes.
We observe that the maximum, per process, number of communication bits in Consensus problem is $\Omega(n)$, therefore one can only hope to improve amortized communication complexity (per process), see Lemma~\ref{lem:worst-case-comm} in Appendix~\ref{sec:auxiliary}.


Our first main result is a quantum algorithm that solves Consensus (correctness with probability~$1$) in expected constant number of rounds and amortized number of qubits and communication bits per process being an arbitrarily low polynomial. This directly improves, by a  polynomial factor, the result of Ben-Or and Hassidim~\cite{Ben-OrH05}, which required a super-linear number of qubits and communication bits, amortized per process. More precisely, in Section~\ref{sec:consensus-alg} we give an algorithm and in Section~\ref{sec:consensus-analysis}~we~prove the following result (see also its technical counterpart -- Theorem~\ref{thm:consensus-constant-technical}):

\begin{theorem}
\label{thm:constant-rounds-consensus}
For 
any $\epsilon>0$, there is an algorithm 
solving consensus against an adaptive $n/3$-adversary in expected
$O(1)$
rounds 
while using $O(n^{\epsilon})$ qubits and communication bits (amortized) per process,~whp.
\end{theorem}

Although our algorithm, called {\sc CheapQuantumConsensus}, uses an idea of so called weak global coin analogously to e.g.,~\cite{FM,Ben-OrH05}, ours is based on two entirely different qubit-and-communication efficient algorithmic techniques.
First technique is a new quantum implementation of a weak global coin, c.f., Definition~\ref{def:coin}, called {\sc CheapQuantumCoin}. It works in constant, $O\big(\frac{1}{\epsilon}\big)$, time, uses an arbitrarily small polynomial number of bits $O(n^{\epsilon})$, amortized per process (for any chosen constant $\epsilon>0$), and results in all non-faulty processes returning the same value with a non-zero 
constant
probability, c.f., Theorem~\ref{thm:coin}.
Second is a deterministic algorithm for counting, which returns, 
in constant time,
a count of the number of input~1 (or 0, resp.) among the processes. The output could be different at different processes, due to failures and a short time for counting, but each of them is a number not smaller than the number of processes with input~1 (0, resp.) at the end of the procedure and not larger than the number of processes with input 1 (0, resp.) in the beginning of the procedure. It uses only an arbitrary small polynomial number of communication bits, amortized per process.
%
We give an overview of these two techniques in the next Sections~\ref{sec:quantum-coin} and \ref{sec:counting}, respectively, and their analysis in Sections~\ref{sec:coin-analysis} and \ref{sec:fast-counting-analysis}, respectively. 

Although constant-time algorithms cannot achieve low communication, as we argue in Lemma~\ref{prop:constant-time} in Appendix~\ref{sec:auxiliary}, interestingly, we show that our main algorithm could be re-instantiated in such a way that it uses only a polylogarithmic number of qubits and communication to solve consensus in a slightly (polylogarithmically) larger number of rounds (see Section~\ref{sec:consensus-alg} and Appendix~\ref{sec:consensus-analysis}, and technical counterpart Theorem~\ref{thm:consensus-polylog-technical} for more details):

\begin{theorem}
\label{thm:polylogarithmic-consensus}
There is an algorithm 
solving consensus against an unbounded~adaptive adversary in 
polylogarithmic number of
rounds, in expectation, while using 
a polylogarithmic number of
qubits and communication bits (amortized) per process,~whp.
\end{theorem}

We believe that the newly developed techniques could be also applied to other types of failures, after failure-specific modifications. For example, although message omission failures require linear amortized communication per process (c.f.,~\cite{HH}), one could still use a small polynomial or even a polylogarithmic number of qubits (together with a linear number of classic communication bits) per process, if qubits are handled according to our techniques while some additional classic communication bits are introduced to handle message omissions. We leave details to follow-up~work.

\subsection{Previous and Related Work}
\label{sec:related}

\paragraph{Consensus in the classic (non-quantum) model.}

Bar-Joseph and Ben-Or~\cite{Bar-JosephB98} (see also their extended version~\cite{Bar-JosephB02}) proved a lower bound $O(\frac{\sqrt{n}}{\log{n}})$ on expected time complexity of consensus against an adaptive adversary.
They also complemented it by time-optimal randomized algorithm. 
Their algorithm uses expected $O(\frac{n^{3/2}}{\log{n}})$ number of communications bits, amortized per process, which has been recently improved by Hajiaghay, Kowalski and Olkowski~\cite{HKO-STOC-2022} to $O(\sqrt{n})$ (while maintaining the almost-optimal round complexity $O(\sqrt{n}\log{n})$).

Fisher and Lynch~\cite{FL} proved a lower bound $f+1$ on {\em deterministic} consensus with $f$ crashes (that actually occurred, i.e., $f<t$), thus separating deterministic solutions from randomized.
Regarding communication complexity, Amdur, Weber and Hadzilacos~\cite{AWH} showed that 
the amortized number of messages
per process is at least constant,
even in some failure-free execution.
Dwork, Halpern and Waarts~\cite{DHW} found a solution with $\cO(\log n)$ messages per process, but requiring an exponential time, and later Galil, Mayer and Yung~\cite{GMY} developed a message-optimal algorithm working in over-linear $\cO(n^{1+\varepsilon})$ time, for any $0<\varepsilon<1$.
They also improved the communication further to a constant number of communication bits per process,
but the resulting algorithm was exponential in the number of rounds.
Chlebus and Kowalski~\cite{CK} showed that consensus can be
solved in  $\cO(f+1)$ time and with $\cO(\log^2 f)$ messages
if only the number $n-f$ of non-faulty processors satisfies $n-f=\Omega(n)$.
It was later improved in~\cite{CK-PODC09} to $O(f+1)$ time and $O(\polylog n)$ number of communication bits.
All the abovementioned communication complexities are amortized~per~process.

\Paragraph{Quantum consensus.} To the best of our knowledge, almost all previous papers on quantum consensus concentrated on assuring feasibility of the problem against strong Byzantine adversaries, 
c.f.,~\cite{Cholvi22,Helm-PODC08,LFZ-Physics19}, or on optimizing time complexity, including the work of Ben-Or and Hassidim~\cite{Ben-OrH05} achieving constant time against an adaptive adversary.
Sparse quantum communication has been considered by Chlebus, Kowalski and Strojnowski in~\cite{ChlebusKS10}, but their protocols work correctly only with some probability smaller than $1$ and for a specific number of failures corresponding to the probability of success. Another difference is that they used quantum operations to encode the classical inputs in quantum registers and used it to directly solve consensus. In this paper, we show another, more-efficient approach, in which we first create a quantum, weak global coin and later employ this tool to the state-of-the-art framework of solving consensus based on the common coin.
Other distributed computing problems, not necessarily fault-prone, were also analyzed in quantum models, c.f.,~\cite{BT-SIGACT08,TKM-TALG12, van2022framework, WuYao}.

\Paragraph{More efficient classic randomized solutions against {\em weak adversaries}.}

Whenever weaker oblivious adversaries are considered,
randomness itself occurred to be enough in reducing time complexity to a constant. 
Chor, Merritt and Shmoys~\cite{CMS} developed constant-time algorithms for consensus against an {\em oblivious adversary} -- that is, the adversary who knows the algorithm but has to decide which process fails and when before the execution starts.
Their solution, however, requires a linear number of communication bits per process.
Gilbert and Kowalski~\cite{GK-SODA-10} presented a randomized consensus algorithm that achieves optimal communication complexity of $O(1) $ amortized communication bits per process 
while terminating in 
$O(\log n) $ time with high probability (whp for short), as long as the number of crashes $ f < n/2 $. 

\Paragraph{Classic consensus in more demanding models.} 

Dolev and Reischuk~\cite{DR} and Hadzilacos and Halpern~\cite{HH} proved the $\Omega(f)$ lower bound on the amortized message complexity per process 
of deterministic consensus for {\em omission or (authenticated) Byzantine failures}. 
However, under some limitation on the adversary and requiring termination only whp, the sublinear expected communication complexity $O(\sqrt{n}\; \polylog{n})$ per process can be achieved even in case of Byzantine failures, as proved by King and Saia~\cite{KingS11}.
Such limitations are apparently necessary
to achieve
subquadratic time complexity for Byzantine failures, c.f., Abraham et al.~\cite{AbrahamCDNP0S19}.
{\em Asynchrony} also implies large communication --
Aspnes~\cite{Aspnes-JACM-98} proved a lower bound $\Omega(n/\log^2 n)$ on communication complexity per process. The complexity bounds in this setting have been later improved, c.f.,~\cite{AC,AlistarhAKS18}. 
\section{Technical preliminaries}
\label{sec:quantum-prel}

\paragraph{Quantum model of computation.}
We focus on quantum properties relevant to our algorithms. A quantum system operates on qubits. Single qubit can be either in a pure or a mixed state. A pure state is a vector in a $2$-dimensional Hilbert space $\mathcal{H}$, while a mixed state is modelled as a probabilistic distribution over pure states. Similarly, a register consisting of $d$ qubits can also be in a pure or a mixed state. A pure quantum state, denoted $\ket{x}$, is a vector of $2^d$-dimensional Hilbert space $\mathcal{H}^{\otimes d} = \mathcal{H} \otimes \ldots \otimes \mathcal{H}$. Again, a mixed state of larger dimension is viewed a probabilistic distribution over pure states. In our paper, we will operate only on pure states, and we will use only the standard computational basis of the Hilbert space, which consists of vectors $\{\ket{b_{1}\ldots b_{d}} : b_{1}\ldots,b_{d} \in \{0,1\}^{d}\}$, to describe the system. Therefore, any state $\ket{x}$ can be expressed as
$\ket{x} = \sum_{i=0}^{2^{d}-1} \alpha_{i}\ket{i}$,
with the condition that $\sum_{i} |\alpha_{i}|^{2} = 1$, since quantum states can be only normalized vectors. 

Transitions, or equivalently -- changes of states of a quantum system, are given by unitary transformations on the Hilbert space of $d$ qubits. These unitary transformations are called {\em quantum gates}. These operations are exhaustive in the sense that any quantum computation can be expressed as a unitary operator on some Hilbert space. There are small-size sets of quantum gates working on two-dimensional space that are universal --
any unitary transformation on a $2^{d}$-dimensional quantum space can be approximated by a finite collection of these universal gates. In our applications, any quantum algorithm computation run by a process requires polynomial 
(in $n$) number of universal~gates.

Finally, an important part of quantum computation 
is also a quantum measurement. Measurements are performed with respect to a basis of the Hilbert space -- in our case this is always the computational basis. A complete measurement in the computational basis executed on a state $\ket{x} = \sum_{i=0}^{2^d - 1}\alpha_{i}\ket{i}$ leaves the state in one of the basis vectors, $\ket{i}$, for $i \in \{0,1\}^{d}$, with probability~$\alpha_{i}^{2}$. The outcome of the measurement is a classic register of $d$ bits, informing 
to which vector the state has been transformed. It is also possible to measure only some qubits of the system, which is called a partial measurement. If $A$ describes the subset of qubits that we want to measure and $B$ is the remaining part of the system, then the partial measurement is defined by the set of projectors $\{ \Pi_{i} = \ket{i}_{A}\bra{i}_{A} \otimes I_{B} \ \, | \ \, \mbox{for } i \in \{0, 1\}^{d}\}$.\footnote{Whenever it could be irrelevant, from the context, we may follow the standard notation in quantum computing and skip writing normalizing factors.} In the former, a subscript refers to the part of the system on which the object exists, $I$ denotes the identity function, while $\bra{i}$ is a functional of the dual space to the original Hilbert space (its matrix representation is the conjugate transpose of the matrix representation of $\ket{i}$).
If before the measurement the system was in a state $\ket{x}_{AB}$ then, after the measurement, it is in one of the states $\{ \Pi_{i} \ket{x}_{AB} \ \, | \ \, \mbox{for } i \in \{0, 1\}^{d}\}$, where state $\Pi_{i} \ket{x}_{AB}$ is achieved with probability $\bra{x}_{AB}\Pi_{i} \ket{x}_{AB}$.\footnote{$\Pi_{i} \ket{x}_{AB}$ and $\bra{x}_{AB}\Pi_{i} \ket{x}_{AB}$ are simply linear operations on matrices and vectors.} The reader can find a comprehensive introduction to quantum computing in~\cite{NielsenQuantum}.

\Paragraph{Graph notations.}
Let $G=(V,E)$ denote an undirected graph.
Let $W\subseteq V$ be a set of nodes of~$G$.
We say that an edge $(v,w)$ of $G$ is \emph{internal for~$W$} if $v$ and~$w$ are both in~$W$.
We say that an edge $(v,w)$ of $G$ \emph{connects the sets~$W_1$ and $W_2$},
or \emph{is between $W_1$ and $W_2$}, for any disjoint subsets $W_1$ and~$W_2$ of~$V$, if one of its ends is in~$W_1$ and the other in~$W_2$.
The \emph{subgraph of~$G$ induced by~$W$}, denoted~$G|_W$, is the subgraph of~$G$ containing the nodes in~$W$ and all the edges internal for~$W$ in $G$.
A node adjacent to a node~$v$ is a \emph{neighbor of~$v$} and the set of all the neighbors of a node~$v$ is the \emph{neighborhood of~$v$}.
$N^i_G(W)$ 
denotes
the set of all the nodes in~$V$ that are of distance at most~$i$ from some node in~$W$ 
in graph~$G$.
In particular, the (direct) neighborhood of~$v$ is denoted~$N_G(v)=N^1_G(v)$. 

The following combinatorial properties are of utter importance 
in the analysis of
our algorithms.
Graph~$G$ is said to be  \emph{$\ell$-expanding}, or to be an \emph{$\ell$-expander}, if any two subsets of~$\ell$ nodes each are connected by an edge.
Graph~$G$ is said to be \emph{$(\ell,\alpha,\beta)$-edge-dense} if, for any set $X\subseteq V$ of \emph{at least} $\ell$ nodes, there are at least $\alpha |X|$ edges internal for~$X$, and for any set $Y\subseteq V$ of \emph{at most} $\ell$ nodes, there are at most $\beta |Y|$ edges internal for~$Y$. 
Graph~$G$ is said to be \emph{$(\ell,\varepsilon,\delta)$-compact} if, for any set $B\subseteq V$ of at least $\ell$ nodes, there is a subset $C \subseteq B$ of at least $\varepsilon\ell$ nodes such that 
each node's degree in~$G|_C$ is at least~$\delta$. 
We call any such set~$C$ a \emph{survival set for~$B$}.

\section{Consensus Algorithm}
\label{sec:consensus-alg}

On a very high level, our consensus algorithm \textsc{CheapQuantumConsensus} repeatedly uses the counting procedure {\sc FastCounting}, specified in Section~\ref{sec:counting}, to compute the number of 0's and 1's preferred by 
processes, see line~\ref{line:fast-counting}. (Recall that the outcomes of {\sc FastCounting} could be slightly different across processes, but by no more than the number of crashes.) Depending on the outcome, each process may change its preferred value to the dominating one (among the received preferred values), decide on it if the domination is substantial, or run the quantum common coin procedure in some almost-unbiased cases -- see lines~\ref{line:decided-one}-\ref{line:decided-zero}~in the pseudocode of \textsc{CheapQuantumConsensus} in Figure~\ref{fig:cons}; all lines involving communication are underlined. 
The latter is a well-established general framework for solving consensus, proposed first by Bar-Joseph and Ben-Or~\cite{Bar-JosephB02} in the context of classic randomized computation against an adaptive adversary. Our novelty is in two new techniques, employed in lines~\ref{line:fast-counting} and \ref{line:quantum-flip} of the pseudocode 
in Figure~\ref{fig:cons}:
fast and quantum/communication-efficient counting 
and 
weak global coin, respectively. Both these techniques use parameters $x, d, \alpha$, which, roughly speaking, correspond to the density of random communication graphs used in these algorithms. 
The detailed performance formulas
of these algorithms, with respect to those parameters, are stated in Theorems~\ref{thm:coin} and~\ref{thm:fuzzy-counting}.

In the heart of these two techniques lies a crucial observation: consensus (as well as common coin and counting) could be achieved 
quickly even if many processes do not directly exchange messages, but use some carefully selected sparse set of communication links instead. This way, 
instead of creating qubits for each pair of processes, we could do it 
only per some pairs corresponding to some communication links to be used. Obviously, this set of links, modeled as an evolving communication graph,
needs to be maintained adaptively and locally by processes throughout the execution -- otherwise, an adaptive adversary would learn it and design crashes to separate processes and prevent consensus. 

\begin{algorithm}
\SetAlgoLined
\SetKwInput{Input}{input}
\Input{$\cP$, $p$, $input_{p}$}
$b_p \leftarrow input_p$ \ \ ; \ \
$r \leftarrow 1$ \ \ ; \ \ $\texttt{decided} \leftarrow FALSE$ \;
\While{$TRUE$}
{
    \ul{%
	participate in \textsc{FastCounting}$(\mathcal{P}, p, b_{p})$ 
	(run with parameters $x = n^{\epsilon}, d = \log{n}, \alpha = n^{\epsilon}$) that counts the processes who have $b_p = 1$ and the processes who have $b_p = 0$; let $O_p^r$, $Z_p^r$ be the numbers of ones and zeros (resp.) returned by \textsc{FastCounting}}\; \label{line:fast-counting}
	$N_p^r \leftarrow Z_p^r + O_p^r$\;
	
	\If{$(N_p^r < \sqrt{n/ \log n })$}
	{
	    \ul{$1)$ send $b_p$ to all processes; $2)$ receive all messages sent to $p$ in round $r+1$}\; 
	    \ul{$3)$ implement a deterministic protocol for $\sqrt{n/ \log n}$ rounds}\;
	}
	
	\If{$\texttt{decided} = TRUE$}
	{
	    \texttt{diff} $\leftarrow$ $N_p^{r-3}$ - $N_p^r$\;
        \lIf{(\texttt{diff} $\leq N_p^{r-2}/10$)}{STOP}
        \lElse{\texttt{decided} $\leftarrow FALSE$}
	}
	\lIf{$O_p^r > (7N_p^r-1)/10$}{$b_p \leftarrow 1$, \texttt{decided} $\leftarrow TRUE$} \label{line:decided-one}
	\lElseIf{$O_p^r > (6N_p^r-1)/10$}{$b_p \leftarrow 1$} \label{line:undecided-one}
	\lElseIf{$O_p^r < (4N_p^r-1)/10$}{$b_p \leftarrow 0$, \texttt{decided} $\leftarrow TRUE$} \label{line:undecided-zero}
	\lElseIf{$O_p^r < (5N_p^r-1)/10$}{$b_p \leftarrow 0$} \label{line:decided-zero}
	\lElse{\ul{set $b_p$ to the output of \textsc{CheapQuantumCoin}$(\mathcal{P}, p)$ executed with parameters $d = \log{n}$, $\alpha = n^{\epsilon}$} }\label{line:quantum-flip}
	
    $r \leftarrow r + 1$\;
}

\Return{$b_{p}$ \tcc*[r]{$p$ outputs final decision}}
\caption{\textsc{CheapQuantumConsensus} for process $p$\label{fig:cons}}
\end{algorithm}

\noindent \textbf{Algorithm's description.} Each process $p$ stores its current choice in $b_{p}$ (which is initialized to $p$'s input). The value $b_{p}$ in the end of the algorithm indicates $p$'s decision. Now, processes use $O(1)$ (in expectation) \textit{phases} to update their values $b_{p}$ such that eventually every process keeps the same decision. 
To do so, in a round $r$ every process $p$ calculates the number of processes whose current choice is $1$ and the number of processes whose current choice is $0$, denoted $O_{p}^{r}$ and $Z_{p}^{r}$ respectively.
Based on these numbers, process~$p$: either sets $b_{p}$ to $1$, if the number $O_{p}^{r}$ is large enough; or it sets $b_{p}$ to $0$, if the number $Z_{p}^{r}$ is large; or it replaces $b_{p}$ with a random bit, if the numbers of zeros and ones are close to each other. If for generating the random bit, in line~\ref{line:quantum-flip}, processes use a quantum implementation of a weak global coin (implemented with \textsc{CheapQuantumCoin} algorithm, specified in Section~\ref{sec:quantum-coin}), they will all have the same value $b_{p}$ with constant probability unless more than third of alive
processes crash. Assuming the presence of the adaptive adversary, this could not be achieved quickly if using
classic communication only. Once it happens with the help of the quantum weak global coin, 
the conditional statements in lines~\ref{line:decided-one}-\ref{line:decided-zero}, 
run in the next iteration of the ``while'' loop, guarantee
that once the majority of processes have the same value $b_p$, the system converges to this value in at most $2$ phases. Since the probability of this event is constant (guaranteed by the quantum weak global coin combined with the analysis of the first classic framework in~\cite{Bar-JosephB02}), the expected number of phases before the consensus algorithm terminates is constant. That reasoning holds, assuming that at most $1 / 3$ fraction of processes crashed (we will generalize it to any $t\le n$ at the end of this section).

As mentioned earlier, the major improvement in the above protocol comes from using novel techniques for counting and weak global coin. For the 
former,
we use the \textsc{FastCounting} algorithm, which, with the choice of parameters given in line~\ref{line:fast-counting}, works in $O\Big(\big(\frac{1}{\epsilon}\big)^{4}\Big)$ rounds and uses $O\big(n^{1 + 3\epsilon}\log^{2}{n}\big)$ (classic) communication bits in total. Similarly, the \textsc{CheapQuantumCoin} algorithm, 
executed in
line~\ref{line:quantum-flip}, terminates in $O\Big(\big(\frac{1}{\epsilon}\big)^{3}\Big)$ rounds and uses $O\Big(n^{1 + 2\epsilon}\log^{2}{n}\Big)$ both quantum and classic bits; we need to divide the 
communication 
formulas by $n$ to obtain the complexity amortized per process. 
By applying these two techniques in the fashion described above, we get Theorem~\ref{thm:constant-rounds-consensus}; detail proof is deferred to Appendix~\ref{sec:consensus-analysis}.

\Paragraph{Handling arbitrary number of crashes.}

Consider $O(\log{n})$ repetitions of the main loop (phases) of the \textsc{CheapQuantumConsensus} algorithm. If during
these phases, the processes with value $b_{p} = 1$ become a large majority (at least $\frac{6}{10}$ fraction of alive
processes), then, as discussed before, every process will decide within the next two rounds. The same holds if processes with value $b_{p} = 0$ start to overpopulate by a ratio of $\frac{6}{10}$ all non-faulty processes. On the other hand, if the cardinalities of the two groups with different values $b_p$ are 
close to each other,
then the processes execute the \textsc{CheapQuantumCoin} algorithm. It outputs a random bit (the same in every participating process), under the assumption that at least a $\frac{2}{3}$ fraction of processes that started this phase as non-faulty have not crashed during this phase. However, in these $O(\log{n})$ phases there must be at least one phase in which the property of a $\frac{2}{3}$ fraction of processes survive holds. In what follows, we argue that if the adversary can crash arbitrarily many processes, but smaller than $n$, 
then the expected number of phases should still be $O(\log{n})$. Now, to obtain the algorithm stated in Theorem~\ref{thm:polylogarithmic-consensus}, we make two more 
adjustments
of the original \textsc{CheapQuantumConsensus} algorithm.
In lines~\ref{line:fast-counting} and \ref{line:quantum-flip}, processes execute the algorithms \textsc{FastCounting} and \textsc{CheapQuantumCoin}, respectively, with parameters $x, d, \alpha$ set as follows: $x = 2, d = \log{n}, \alpha = \log{n}$. 
This corresponds to a use of sparse graph for communication (of degree roughly $O(\log{n})$). In consequence, the time complexity of the \textsc{FastCounting} algorithm increases to $O(\log^{4}{n})$, but the communication complexity decreases to $O(\log^{8}{n})$ amortized per process.
The details are presented in Section~\ref{sec:counting}, and performance follows from putting the abovementioned parameters to the formulas in Theorem~\ref{thm:fuzzy-counting}). Similarly, the time complexity of the \textsc{CheapQuantumCoin} algorithm increases to $O(\log^{3}{n})$, but the communication complexity (both quantum and classical) decreases to $O(\log^{7}{n})$ amortized per process.
The details are presented in Section~\ref{sec:quantum-coin}, and performance follows from putting the abovementioned parameters to the formulas in Theorem~\ref{thm:coin}).

\section{Qubit-and-Communication Efficient Quantum Common Coin}
\label{sec:quantum-coin}

In this section, we design a new 
$t$-resilient weak global coin, for $ t < n$, with the help of quantum communication and computation. 
The coin must satisfy the following definition:

\begin{definition}[\cite{Ben-OrH05}] 
\label{def:coin}
Let $\cC$ be a protocol for $n$ players (with no input), where each player $i$ outputs a (classical) bit $v_i \in \{0, 1\}$. We say that the protocol $\cC$ is a $t$-resilient weak global coin protocol (or computes a weak global coin, for short) with fairness $\rho > 0$, if for any adaptive $t$-adversary and any value $b \in \{0,1\}$, with probability at least $\rho$, $v_i = b$ for all good players $i$.
\end{definition}

On the high level, our protocol {\sc CheapQuantumCoin} chooses a leader process uniformly at random and all other processes agree on the random bit proposed by the leader. Quantum phenomena is used to hide the random choices of the leader and its output from the adaptive adversary when processes communicate with each other. The idea was first proposed in~\cite{Ben-OrH05}, yet there are key differences between that work and our algorithm. Instead of all-to-all communication, which required large number of qubits, we use a sequence of random graphs of node degrees $d, d\alpha^{1}, \ldots, d{\alpha}^{k}$, respectively, where $d,\alpha \in \Omega(\log{n})$ and $k = \ceil{\log{n} / \log{\alpha}}$ 
are some integer parameters. The vertices of these graphs correspond to processes and edges correspond to communication link -- each process communicates with neighbors in one of the graphs at a time. 
If the graph is chosen properly (i.e., so that there is no situation in which too many processes use denser graphs), it reduces the communication complexity, but simultaneously imposes a new challenge.
Mainly, the communication procedure has to now assure delivery of quantum bits between every two non-faulty processes regardless of the pattern of crashes. For instance, if only one random graph of degree $d$ was used then the adversary could easily isolate any vertex using only $O(d)$ crashes (i.e., by crashing all its neighbors). 
Hence, strictly speaking, assuring such delivery is not possible while using a sparse communication graph as relays, but we show that a certain majority could still come up with a common coin value based only on their exchanges with neighbors in the communication graphs; they could later propagate their common value to other processes by adaptively controlling their (increasing) set of neighbors, taken from subsequent communication graphs of increasing density.
A thorough analysis shows that in this way it is possible to achieve the same quantum properties
that are guaranteed by
Ben-Or's and Hassidim's global coin~\cite{Ben-OrH05}, and at the same time reducing the quantum communication by a polynomial factor. Formally, we prove the following~result.

\Paragraph{Algorithm's description.}
We now describe the 
{\sc CheapQuantumCoin} algorithm. Its pseudo\-code
is presented in Figure~\ref{fig:quantum-coin}. It takes as an input: a process name $p$ and two integers, $d$ and $\alpha$. The two latter parameters are used to determine communication pattern between non-faulty~processes, and their choice determines complexity of the algorithm.

\jo{


}

As mentioned before, processes use quantum equivalent of a procedure in which processes draw random names from the range $[1, \ldots, n^{3}]$ and decide on a random bit proposed by the process with the largest name.\footnote{%
Note that the latter procedure cannot be used against an adaptive adversary, as it could crash such a leader.}
We view the quantum part of the algorithm as a quantum circuit on the joint space of all qubits ever used by different processes. 
Due to the distributed nature of the system, not all quantum operations are allowed by the quantum circuit. That is,
(i) any process can perform arbitrary unitary gates on its part of the system, (ii) unitary gates on qubits of different processes might be performed, but must be preceded by quantum communication that send qubits involved in the operation to a single process. The communication can be either direct or via relays. 
We next explain what unitary gates can be used to simulate the classic algorithm of the leader election in the quantum distributed setting. For the purpose of explanation, we assume that the qubits needed to execute each of the following gates have been delivered to proper processes. The pattern of communication assuring this delivery
is described in the next paragraph.
\begin{enumerate}
    \item \texttt{Hadamard gate.} This unitary operation on a single qubit is given by the matrix $H = \frac{1}{\sqrt{2}}\begin{bmatrix} 1 & 1 \\ 1 & -1  \end{bmatrix}$. When applied to every qubit of an $m$-qubit state $\ket{0\ldots0}$, it produces a uniform superposition of vectors of the computational basis, i.e.,
    $$H^{\otimes m} \ket{0\ldots0} = \frac{1}{\sqrt{2} \cdot n^{3/2}} \sum_{\ell_{1},\ldots,\ell_{m - 1} \in \{0,1\}^{m - 1}, b \in \{0,1\}} \ket{\ell_{1}\ldots \ell_{m} b} = \frac{1}{2^{m/2}} \sum_{i = 0}^{2^m-1} \ket{i} \ .$$
    In the beginning, every process applies this gate to its main register $\ket{Leader}\ket{Coin}$ consisting of $3\log + 1$ qubits, (line~\ref{line:init}). This way the entire system is in the following state:
    $$\left(\frac{1}{\sqrt{2} \cdot n^{3/2}} \sum_{i = 0}^{2n^3 - 1} \ket{i} \right)^{\otimes n} = 
    \left(\frac{1}{\sqrt{2} \cdot n^{3/2}} \sum_{i = 0}^{2n^3 - 1} \ket{i} \right) \otimes \ldots \otimes \left(\frac{1}{\sqrt{2} \cdot n^{3/2}} \sum_{i = 0}^{2n^3 - 1} \ket{i} \right).$$
    Observe, that if all processes measure the main registers in the computational basis, they the outcome has the same probability distribution as a random experiment in which every process draws a number from uniform distribution over $[1, \ldots, n^3]$ (first $3\log{n}$ qubits of each register) and a uniform random bit (the last qubit of each register).

    \item \texttt{CNOT} and \texttt{Pairwise\_CNOT.}
    The $CNOT$ is a 2-qubit gate given by the following unitary transform
    $$CNOT\left(\alpha\ket{00} + \beta\ket{01} + \gamma\ket{10} + \delta\ket{11}\right) =  \alpha\ket{00} + \beta\ket{01} + \gamma\ket{11} + \delta\ket{10}.$$
    The $\texttt{Pairwise\_CNOT}$ gate works on a register of $2\cdot m$ qubits. Let $A$ be the first half of the register while $B$ is the second. The gate applies the $CNOT$ gate qubit-wisely to corresponding pairs of qubits of $A$ and $B$. We only use this gate when the part $B$ of the entire register
    is in the state $\ket{0\ldots0}$. 
    If the part $A$ of the enitre register is in a state $\sum_{i = 0}^{2^m-1} \alpha_{i}\ket{i}_{A}$ then the unitary gate results in
    $$\texttt{Pairwise\_CNOT}\left(\sum_{i = 0}^{2^m-1} \alpha_{i}\ket{i}_{A}, \ket{0\ldots0}_{B}\right) =   \sum_{i = 0}^{2^m-1} \alpha_{i}\ket{i}_{A}\ket{i}_{B} \ .$$
    Whenever a process decides to send its qubits to another process,
    it performs the $\texttt{Pairwise\_CNOT}$ gate on its qubits and the new block of $3\log{n} + 1$ qubits initialized to $\ket{0\ldots0}$. Since, as a result, the information needed to be propagated is now also encoded on another $3\log{n} + 1$ qubits, the processes can send the new block to another part keeping the original qubits for itself. Observe, however, that this is not a universal copy gate (which is not possible in quantum model of computation), since it does not copy the part $A$ to the part $B$, but rather achieves a specific form of entanglement that encodes the same information if described in the computational~basis.
    
    \item \texttt{F\_CNOT} and \texttt{Controlled\_Swap}.
    Let $f : \{0,1\}^{m} \rightarrow \{0, 1\}$ be a Boolean function. Let $A$ be a register of $m$ qubits and $C$ be an additional single qubit register. The gate $\texttt{F\_CNOT}_{f}$ performs $NOT$ operation on the last qubit iff $f(A) = 1$. If $C = \ket{0}$, then we get 
    $$\texttt{F\_CNOT}_{f}\left(\sum_{i = 0}^{2^m-1} \alpha_{i}\ket{i}_{A}, \ket{0}_{C}\right) = \sum_{i = 0}^{2^m-1} \alpha_{i}\ket{i}_{A}\ket{f(i)}_{C} \ .$$
    The gate can be implemented using $2^{m}$ unitary transforms. For each vector $x$ from the computational basis such that $f(x) = 1$, we apply $NOT$ on the last qubit and then compose at most $2^{m}$ such gates. We note here that we never use the $\texttt{F\_CNOT}$ gate on registers of more than $3\log{n} + 1$ qubits, therefore the \texttt{F\_CNOT} has polynomial size.
    
    The $\texttt{Controlled\_Swap}$ operates on a control register $C$ and two registers of the same size, $A$~and~$B$. It swaps the content of two registers of the same size if the control register is $\ket{1}$. Formally, the gate works as follows: 
    $$\texttt{Controlled\_Swap}\left(\ket{i}_{A}, \ket{j}_{B}, \alpha\ket{0} + \beta\ket{1}\right) =  \alpha \cdot \ket{i}_{A}\ket{j}_{B}\ket{0} + \beta \cdot \ket{j}_{A}\ket{i}_{B}\ket{1} \ .$$
    
    We use these two gates whenever a process receives qubits of another process. First, a control qubit $\ket{S}$ is prepared by comparing the content of the received register and the main register of the process. The condition function used in the gate is the following:
    $$f : \{0,1\}^{3\log{n} + 1} \times \{0,1\}^{3\log{n} + 1} \rightarrow \{0, 1\} \ \ , \ \  f(i, j) \iff i > j \ .$$ 
    Next, the qubit $\ket{S}$ is passed to the $\texttt{Controlled\_Swap}$ gate, operating again on the main register and the newly received qubits. For the $\ket{1}$ part of the control qubit, which corresponds to the fact that the received register has larger values in the computational basis, the gate switches the content of the two registers.
\end{enumerate}

Let us now explain how all the gates listed above work together. As mentioned in the description of the Hadamard gate, after line~\ref{line:init} ends, the main registers of the system are in the state being a uniform superposition of all vectors of the computational basis.
Starting from this point, the composition of all gates applied to different registers along the execution can be viewed as a single unitary gate on the entire system, consisting of the qubits that any processes ever created. Note that the unitary transformation might be different depending on the failure in communication, i.e., a failure in delivery of some block of qubits between two processes may result in abandoning gates involving these qubits, but for 
a moment let us assume that the links are reliable. Since the unitary transformation is linear, it is enough to consider how it affects the vectors of the computational basis. However, all the gates described above behave in the computational basis as their classic equivalents. More precisely, let $\ket{x}$ be a vector from the computational basis spanning the whole circuit. Let $p$ be the process whose main register has the largest\footnote{The probability of a tie is polynomially small} value in the state $\ket{x}$. From the point of view of this register, the following happens in the algorithm. In each round, $p$ creates an entangled state on $6\log{n} + 2$ qubits (see point $(2)$) that has the same qubits on its new block of $3\log{n} + 1$ qubits as it has on the main register. Then, it propagates the new block to its neighbors (line~\ref{line:copy}). The neighbors compare the content of received qubits and exchange them with their main register if their content is smaller (gates $\texttt{F\_CNOT}$ and $\texttt{Controlled\_Swap}$ in lines~\ref{line:quantum-flip}). 
This operation is then repeated
$(k + 2)^{2}(\gamma_{\alpha} + 1)$ on the set of links defined by some random evolving graphs, see the later paragraph about adaptive communication pattern. 
In the end, the processes who, either directly or via relays, received the content of the largest main register, have the same value in their main register. Therefore, the result of the measurement in line~\ref{line:measurement} must be the same in all these~processes.  

Assume now that we are able to achieve bidirectional quantum communication between any pair of processes of an $\alpha$-fraction of the entire system, regardless of the (dynamic) actions of the adversary. From the above, the algorithm transforms any vector whose largest main register is one of the registers of the $\alpha$-fraction to a vector such that processes from the $\alpha$-fraction have the same values in main registers. Connecting the above discussion with the fact the algorithm is a $1-$to$-1$ transformation on the linear space of states, we get that the probability of measuring the same values in the processes of the $\alpha$-fraction is at least the probability of measuring the largest value in one of the processes belonging to the $\alpha$-fraction, if the measurement was done before the quantum communication. Since the state is a uniform superposition at that time, 
the probability is at least $\alpha - o(1)$ and we can claim the following lemma.\footnote{The $o(1)$ part contributes to the unlikely event that there are ties between largest values.}

\begin{lemma}\label{lem:chance-win}
Let $A$ be a set of correct processes such that any pair of them was connected by quantum communication either directly or by relays. Then the probability that all processes from $A$ output the same bit from the algorithm $\textsc{CheapQuantumConsensus}$ is at least $\frac{|A|}{n} - o(1)$. 
\end{lemma}

\SetAlFnt{\small}
\begin{algorithm}
\SetAlgoLined
\SetKwFor{ForAll}{for each}{do (in one communication round)}{end}
\SetKwInput{Input}{input}
\SetKwInput{Output}{output}
\Input{$p$, two parameters: $d, \alpha$}

For $0 \le i \le \ceil{\frac{\log(n / d)}{ \log{\alpha}}}$: $\mathcal{N}_{p}(d\alpha^{i}) \leftarrow $ a set of processes such that each process is chosen independently with probability $\frac{d\alpha^{i}}{n}$ \;
$\texttt{degree}_p \leftarrow d$, \hspace{1.5mm} $\gamma_{\alpha} \leftarrow \frac{\log{n}}{\log{\alpha}}$, \hspace{1.5mm} $\delta_{\alpha} \leftarrow \frac{2}{3}\log{n}$  \label{line:init}\;

$\ket{Leader}_{p}\ket{Coin}_{p} \leftarrow H^{\otimes{n}}\ket{00\ldots0}$ (a gate on $3\log{n} + 1$) \; 


\For{$i = 1$ to $(k+2)^{2}$
\label{line:epochs} \tcc*[r]{iter. of epochs}}
{
    $\texttt{adaptive\_degree} \leftarrow \texttt{degree}_{p}$  \;\label{line:copy}
    \For{$j = 1$ to $\gamma_{\alpha} + 1$ \label{line:testing} \tcc*[r]{iter. of testing rounds}}
    {
           \underline{send to each process in $\mathcal{N}_{p}(\texttt{degree}_{p})$: an inquire bit $1$} \; \label{line:inquires}
        $\mathcal{I} \leftarrow$ the set of processes who sent an inquire bit to $p$ \;
        $\forall_{q \in \mathcal{I}} : \ket{B_{q}} \leftarrow \texttt{Pairwise\_CNOT}\left(\ket{LeaderCoin}_{p}, \ket{0\ldots0}\right)$ \;
        \underline{send to each process $q \in \mathcal{I}$: a quantum message containing first $\log^3 n$ bits of $\ket{B_{q}}$,} \underline{and a classical message containing $\texttt{adaptive\_degree}_{p}$ } \; \label{line:responses}
        \BlankLine
        \BlankLine
        $\mathcal{R} \leftarrow$ the set of processes who responded to $p$'s inquires \;
        \ForEach{$q \in \mathcal{R}$}
        {
            $\ket{CLeader}_{q}\ket{CCoin}_{q} \leftarrow$ received quantim bits from $q$, \hspace{1mm} $\ket{S} \leftarrow \ket{0}$\;
            $\texttt{F\_CNOT}_{Leader_{p} > CLeader_{q}}\left(\ket{Leader}_{p}, \ket{Leader}_{q}, \ket{S} \right)$\; \label{line:quantum-updates-begin}
            $\texttt{Controlled\_Swap}\left(\ket{LeaderCoin}_{p}, \ket{CLeaderCCoin}_{q}, \ket{S}\right)$\; \label{line:quantum-updates-end}
        }

        \BlankLine
        \BlankLine
        \While{\label{line:while}$|\{q \in \mathcal{R} : \texttt{adaptive\_degree}_{q} \ge \texttt{adaptive\_degree}_{p} \}| < \delta_{\alpha}$ 
        and $\texttt{adaptive\_degree}_{p} \ge d$ 
        \tcc*[r]{adapting 
        \#neighbors
        during testing} \label{line:while-coin}}
        {
            $\texttt{adaptive\_degree}_{p} \leftarrow \frac{1}{\alpha} \texttt{adaptive\_degree}_{p}$ \;
        }
    }
    \uIf{\label{line:if}$\texttt{adaptive\_degree}_{p} < \texttt{degree}_{p}$}
    {
        $\texttt{degree}_{p} \leftarrow \min\{\texttt{degree}_{p} \cdot \alpha, d\alpha^{k}\}$ 
        \tcc*[r]{neighborhood~for~next~epoch grows}
    }
}
$b_{p} \leftarrow$ be the last bit the measurement of $\ket{Leader}_{p}\ket{Coin}_{p}$ in the computational basis\;\label{line:measurement}
\Return{$b_{p}$ 
    \tcc*[r]{$p$ outputs random bit}}
\caption{\textsc{CheapQuantumCoin} for process $p$ \label{fig:quantum-coin}
}
\end{algorithm}

\Paragraph{Adaptive communication pattern.} 
As explained, 
we not only need that the communication should 
be efficient in terms of the number of qubits and classic bits, but also it should be such that any two correct processes of a large fraction of the entire system are connected by a short path of correct process so that quantum registers can be relayed.
Let $d, \alpha$ be two integers parameters. We define $k = \ceil{\frac{\log(n / d)}{ \log{\alpha}}}$, $\gamma_{\alpha} = \frac{\log{n}}{\log{\alpha}}$, and $\delta_{\alpha} = \frac{2}{3}\alpha$. Initially, each process $p$ draws independently $k + 1$ sets $\mathcal{N}_{p}(d), \mathcal{N}_{p}(d\alpha^{1}), \ldots, \mathcal{N}_{p}(d\alpha^{k})$, where a set $\mathcal{N}_{p}(d\alpha^{i})$, for $0 \le i \le k$, includes each process from $\cP$ with probability $\frac{d\alpha^{i}}{n}$.


Communication is structured into $(k + 2)^{2}$ \textit{epochs}, see line~\ref{line:epochs}. Each epoch consists of $2(\gamma_{\alpha} + 1)$ communication rounds, also called \textit{testing} rounds. They are scheduled in $\gamma_{\alpha}+1$ iterations within the loop ``for'' in line~\ref{line:testing}, each iteration containing two 
communication
rounds (underlined in the pseudocode): sending/receiving inquiries in line~\ref{line:inquires} and sending/receiving responses in line~\ref{line:responses}. In the testing rounds of the first epoch, a process $p$ 
sends inquiries to processes in set $\mathcal{N}_{p}(d)$.
The inquired processes respond by sending in the next round (line \ref{line:responses}) their current classic state and specially prepared, in line~\ref{line:copy}, quantum register.
However, if in a result of crashes $p$ starts receiving less than $\delta_{\alpha}$ responses per round, it switches its communication neighborhood from $\mathcal{N}_{p}(d)$ to the next, larger set,
$\mathcal{N}_{p}(d\cdot \alpha)$. A similar adaptation to a crash pattern is continued in the remaining~epochs. 

Process $p$ stores the cardinally of the set being inquired in an epoch in the variable $\texttt{degree}_{p}$ (initialized to $d$ in line~\ref{line:init}). For the purpose of testing rounds, $p$ copies the value $\texttt{degree}_{p}$ to a variable $\texttt{adaptive\_degree}_{p}$ (line~\ref{line:copy}). In every testing round, $p$ adapts its variable $\texttt{adaptive\_degree}_{p}$ to the largest value $x \le \texttt{adaptive\_degree}_{p}$ such that it received at least $\delta_{a}$ responses from processes that have their variable $\texttt{adaptive\_degree}$ at least $x$ (loop ``while'' in line~\ref{line:while-coin}). If $p$ had to decrease the value $\texttt{adaptive\_degree}_{p}$ in testing rounds of an epoch, it then \textit{increases} the main variable $\texttt{degree}_{p}$ by the factor $\alpha$ before the next epoch, see line~\ref{line:if}. The latter operation formally encodes
the intuition that the process $p$ expected to have $\delta_{\alpha}$ non-faulty neighbors with their values of $\texttt{degree}$ at least as big as its own, but due to crashes it did not happen; Therefore, $p$ increases the 
number
of inquired processes, by adopting the next, larger neighborhood set $\mathcal{N}_{p}(\cdot )$, randomly selected, in order to increase the chance of 
communication with the majority of non-faulty processes in the next epoch. On the other hand, the adaptive 
procedure
of reducing $\texttt{adaptive\_degree}$ in testing rounds of a single epoch helps neighbors of $p$ to estimate correctly 
the size of the
neighborhood that process $p$ is using in the current testing round, which might be much smaller than the value $\texttt{degree}_{p}$ from the beginning of the epoch. This, in turn, calibrates the value of $\texttt{adaptive\_degree}$ of the neighbors of $p$, and this calibration can propagate to other processes of distance up to $\gamma_{\alpha}$ from $p$ in the next iterations of testing rounds.
%

\Paragraph{Analysis.}
Let us define graphs $\mathcal{G}(d\alpha^{i})$, for $0 \le i \le k$, as the union of random sets $\cup_{p \in \mathcal{P}} \mathcal{N}_{p}(d\alpha^{i})$. The probability distribution of the graph $\mathcal{G}(d\alpha^{i})$ is the same as the random graph $G(n, y)$ for $y = \frac{d\alpha^{i}}{n}$. Chlebus, Kowalski and Strojnowski~\cite{CK-PODC09} 
showed in their Theorem~2, applied for $k = \frac{64n}{d\alpha^{i - 1}}$,
that the graph $\mathcal{G}(d\alpha^{i})$ has the following properties, whp:
\vspace{0.5em}
\\
\emph{(i)} it is $(\frac{n}{d\alpha^{i - 1}})$-expanding, which follows from $(\frac{n}{d\alpha^{i - 1}}, \frac{2}{3}\frac{n}{d\alpha^{i - 2}}, \frac{4}{3}\frac{n}{d\alpha^{i - 2}})$-edge-expanding property, 
\\
\emph{(ii)} it is $(\frac{n}{d\alpha^{i - 1}},\frac{1}{3}\alpha,\frac{2}{3}\alpha)$-edge-dense, 
\hspace*{4.8em}
\emph{(iii)} it is $(16\frac{n}{d\alpha^{i - 1}},3/4,\frac{2}{3}\alpha)$-compact,
\\
\emph{(iv)} the degree of each node is at most $\frac{21}{20}d\alpha^{i}$.

Since the variable $\texttt{degree}_{p}$ takes values in the set $\{d, d\alpha^{1}, \ldots, d\alpha^{k}\}$, 
the pigeonhole principle assures that eventually $p$ will participate in an epoch in which $\texttt{degree}_{p}$ has not been increased (in the most severe scenario, $p$ will use the set $\mathcal{N}_{p}(d\alpha^{k})$, which consists of all other processes -- 
because it
contains each 
process, as a neighbor of $p$,
with probability $1$). The \underline{$(\gamma_{\alpha}, \delta_{\alpha})$-dense-neighborhood property} of random graphs composed from neighborhoods $\mathcal{N}(\texttt{degree}_{p})$ implies that 
$p$ will then contact a majority of other non-faulty processes
via at most $\gamma_{\alpha}+1$ intermediate processes (during testing rounds). Formally, the following holds:

\begin{lemma}\label{lem:surviving-probing}
If a process $p$ does not change its variable $\texttt{degree}_{p}$ at the end of an epoch $i$, then at the beginning of epoch $i$ there exists a $(\gamma_{\alpha}, \delta_{\alpha})$-dense-neighborhood of $p$ in the graph $\mathcal{G}(\texttt{degree}_{p})$ consisting of non-faulty processes 
with
variable $\texttt{degree}$ 
being
at least $\texttt{degree}_{p}$ in the epoch $i$, whp.
\end{lemma}


On the other hand, \underline{$(16n / d\alpha^{i - 1}, 3/4, 2\alpha / 3)$-compactness} of the (random) graph composed of processes that have the variable $\texttt{degree}$ at least $d\alpha^{i}$, guarantees that the total number of processes that use sets $\mathcal{N}(d\alpha^{i})$ during the epoch $i$ is at most $\frac{n}{\alpha^{i - 2}}$, which amortizes communication complexity.

\begin{lemma}\label{lemma:small-sets}
For any integer $i$, such that $0 \le i \le k$, at the beginning of each epoch there is at most $\frac{16n}{d\alpha^{i - 2}}$ processes with the variable $\texttt{degree}$ greater than $d\alpha^{i}$, whp.
\end{lemma}

The above discussion yields the following result.

\begin{theorem}\label{thm:coin}
For two integer parameters $d, \alpha \in \Omega(\log{n})$, 
the algorithm \textsc{QuantumCoinFlip} generates a weak global coin, 
provided
that at most $\frac{1}{3}$-fraction of initially non-faulty processes have crashed. It terminates in $O\Big(\big(\frac{\log{n}}{\log{\alpha}}\big)^{3}\Big)$ rounds and with high probability uses $O\Big(\big(\frac{\log{n}}{\log{\alpha}}\big)^{4} d\alpha^{2}\log{n}\Big)$ both quantum and classic communication bits (amortized) per process.
\end{theorem}

\section{Constant-Time Communication-Efficient Fuzzy Counting}
\label{sec:counting}

Although generating a weak global coin in a constant number of rounds against an adaptive adversary requires quantum communication (due to the lower bound by Ben-Or and Bar-Joseph~\cite{Bar-JosephB98}), the \textsc{CheapQuantumCoin} algorithm, even without quantum communication, achieves few other goals. As discussed in the previous section, its random communication pattern guarantees, whp, that any additional classic message, also called a rumor, of a non-faulty process can be conveyed to any other non-faulty process if 
added to every classic message 
sent/received in
line~\ref{line:responses}. Even more, assume that there is a set of $x$ messages/rumors $M = \{ m_{1}, \ldots, m_{x} \}$, 
distributed as inputs among some subset of processes (one message from $M$ per process). If processes always convey all the known  
rumors
from set $M$ when using classic communication (avoiding repetition), then they solve a Gossip problem, whp, i.e., every 
rumor
$m_{i}$ given to a non-faulty process, for $1 \le i \le x$, is known at the end of the protocol to every other non-faulty process. Observe that processes resign from the quantum content of communication for this purpose, and instead of $\log{n}$ bits (or qubits) per message, they use $|M|$ bits, where $|M|$ denotes the number of bits needed to encode all 
rumors
from $M$. Finally, if processes work in a model where the names of other processes are commonly known, they can withdraw from using random communication. Instead, they can use a deterministic family of graphs $\mathcal{G}(d\alpha^{i})$, for the same choice of parameters $d$ and $\alpha$. The proof of existence of such graphs, using the probabilistic argument, was described 
in \cite{CK-PODC09} (Theorem~2). In 
such
case, the set $\mathcal{N}_{p}(d\alpha^{i})$ is the neighborhood of process $p$ in the deterministic graph $\mathcal{G}(d\alpha^{i})$. 
(Processes compute the same copies of graphs $\mathcal{G}$ locally in the beginning of the algorithm.)
The above augmentation of the algorithm,
together with the proof of
Theorem~\ref{thm:coin}, from which we take the upper bound on the number of messages send and the upper bound on the number of rounds, leads to the following result.

\begin{theorem}
\label{thm:fast-gossip}
For 
integer parameters $d, \alpha \in \Omega(\log{n})$, there is a deterministic algorithm that solves the gossip problem in $O\left(\left(\frac{\log{n}}{\log{\alpha}}\right)^{3}\right)$ rounds using $O\left(\left(\frac{\log{n}}{\log{\alpha}}\right)^{4} d\alpha^{2}\cdot (|M| + \log{n})\right)$ communication bits per process (amortized), where $|M|$ is the number of bits needed to encode all~rumors~in~$M$.
\end{theorem}

\Paragraph{Generalized Fuzzy Counting.}

\begin{definition}[Fuzzy Counting~\cite{HKO-STOC-2022}]
An algorithm solves Fuzzy Counting if
each process returns a number between the initial and the final number of active processes. Here, the notion of ``being active'' depends on the goal of the counting, e.g., all non-faulty processes, processes with initial value $1$, etc.
\end{definition}

\noindent Note that the returned numbers could be different across processes.
Here, we refine the state-of-art method of solving
the fuzzy counting problem, even deterministically, and propose a new recursive algorithm 
with
any branching factor $x$, called $\textsc{FastCounting}$. 
Formally, we prove the following:


\begin{theorem}
\label{thm:fuzzy-counting}
For any $2 \le x \le n$ and $\delta, \alpha \in \Omega(n)$, the \textsc{FastCounting} deterministic algorithm solves the \Fuzzy Counting problem in $O\Big(\frac{\log{n}}{\log{x}}\big(\frac{\log{n}}{\log{\alpha}}\big)^{3}\Big)$ rounds, using $O\Big(\frac{\log{n}}{\log{x}}\big(\frac{\log{n}}{\log{\alpha}}\big)^{4} d\alpha^{2} \cdot x\log{n}\Big)$ communication bits (amortized) per process.
\end{theorem}

Obviously, the constant-time is achieved in Theorem~\ref{thm:fuzzy-counting} when $x, \alpha = n^{\epsilon}$, for a constant $\epsilon \in (0,1)$. In this case, the number of rounds is $O\Big(\big( \frac{1}{\epsilon} \big)^{4}\Big)$, while the communication complexity is $O(n^{3\epsilon}\log^{2}{n})$ (amortized) per process. In our approach, we generalize the method of Hajiaghayi et al. 
\cite{HKO-STOC-2022} to denser communication graphs of certain properties, which allows us to achieve constant running time. The constant running time is the key feature of the algorithm, since it is used in the implementation of (expected) constant-round quantum consensus protocol {\sc CheapQuantumConsensus}. The main difference between ours and the state-of-art approach is a different Gossip protocol 
used in the divide-and-conquer method. 

The $\textsc{FastCounting}$ algorithm is recurrent. It takes as 
an
input the following values: $\mathcal{P}$ is the set of processes on which the algorithm is executed; $p$ is the 
name of a process which executes the algorithm; $a_{p} \in \{0, 1\}$ denotes if $p$ is active ($a_p=1$) or not;
and parameters $x, d, \alpha$, where $x$ is the branching factor and $d, \alpha$ steer the density of the  communication graph in the execution. Knowing the set $\mathcal{P}$ of $n$ processes, $\textsc{FastCounting}$ splits the set into $x$ disjoint groups of processes, each of size between $\floor{\frac{n}{x}}$ and $\ceil{\frac{n}{x}}$. Name these groups $\mathcal{P}_{1}, \ldots, \mathcal{P}_{x}$. The groups are known to each participating process. The algorithm then makes $x$ parallel recursive calls on each of these groups. As a result, a process $p$ from a group $\mathcal{P}_{i}$, for $1 \le i \le x$, gets the number of the active processes in its group $\mathcal{P}_{i}$. In the merging step, all processes execute Gossip algorithm, proposed in Theorem~\ref{thm:fast-gossip}, with parameters $d, \alpha$, where the input rumors are the numbers calculated in the 
recursive
calls. To keep the communication small, when processes learn new rumors they always store at most one rumor corresponding to each of the $x$ groups. This way, the number of bits needed to encode all rumors is $O(x\log{n})$. 
Let $r_{1}, \ldots, r_{\ell}$ be the rumors learned by process $p$ from the execution of the Gossip algorithm. The final output of $p$ is the sum $\sum_{i = 1}^{\ell} r_{i}$. We provide the pseudocode of the algorithm in Figure~\ref{fig:counting}. 
Note that all processes run the algorithm to help counting and Gossip, not only those with $a_p=1$, but only the latter are counted (see lines 2, 5, 7).
For detail analysis of the round and communication bit complexity we refer to Section~\ref{sec:fast-counting-analysis}.
\SetAlFnt{\small}
\begin{algorithm}[ht!]
\SetAlgoLined
\SetKwInput{Input}{input}
\Input{$\cP$, $p$, $a_{p}$; $x, d, \alpha$}
\uIf{\label{line:base-case} $|\cP| = 1$}
{
    \Return{$a_{p}$}
}
$\mathcal{P}_{1}, \ldots, \mathcal{P}_{x} \leftarrow$ partition of $\mathcal{P}$ into $x$ disjoint group of size between $\floor{\frac{|\mathcal{P}|}{x}}$ and $\ceil{\frac{|\mathcal{P}|}{x}}$ 
\tcc*[r]{this partition is common for all processes, since we assumed the common knowledge of set $\mathcal{P}$}
let $g$ be the index of the $p$'s group, i.e., $p \in \mathcal{P}_{g}$\; 
$\texttt{active}_{p} \leftarrow \textsc{FastCounting}(\mathcal{P}_{g}, p, a_{p}, x, d, \alpha)$\; \label{line:recursive-counting}
$\mathcal{R} \leftarrow$ the outcome of the Gossip algorithm, referred to in Theorem~\ref{thm:fast-gossip}, executed on the set of processes $\mathcal{P}$ with initial message $\texttt{active}_{p}$ 
and parameters $d, \alpha$\; \label{line:fast-gossip}

\Return{\label{line:counting-output} $\sum_{q \in \mathcal{R}} \texttt{active}_{q}$}
\caption{\textsc{FastCounting} for process $p$\label{fig:counting}}
\end{algorithm}






\newpage

\bibliographystyle{apa}
\bibliographystyle{ACM-Reference-Format}

\appendix
\begin{center}
    {\Large\bf APPENDIX}
\end{center}

\section{The Analysis of CheapQuantumConsensus Algorithm}
\label{sec:consensus-analysis}

To analyze the \textsc{CheapQuantumConsensus} algorithm we first recall a combinations of lemmas from \cite{Bar-JosephB02}.

\begin{lemma}[Lemmas $4.1, 4.2, 4.3$ in \cite{Bar-JosephB02}]\label{lem:same-b}
If all processes have the same value at the beginning of an iteration of the main while loop, then the algorithm returns the decision after at most two iterations.
\end{lemma}

\remove{
The correctness proof in Theorem 7 is incomplete. The protocol in [6] only tolerates $n/10$ corruption. Moreover, the authors also modifies the protocol in [6]. Thus, the correctness does not directly come from [6]

Response: We apologize for the ambiguity. Since there is a version of [6] that describes a protocol tolerating $t < n$ number of faults we decide to refer rather to this paper than applying a borrowed analysis in our paper. This relates also to the modifications whose influence is discussed in the proof of Theorem 7. Specifically, we show that any output of our quantum coin could be viewed as an output of the regular common coin with non-zero probability. Since the algorithm from [6] works with probability 1, any output of the new quantum coin cannot affect the correctness of the algorithm. However, we understand that these gaps might require more detailed analysis and we will be happy to include the full proof of correctness in the next version of our paper.  

!!!Tutaj nie wiem jak poprawic dowod Theorem 7. Moim zdaniem to jest czytelnie napisane jakie zmiany mamy w algorytmie.!!!

}

\begin{theorem}
\label{thm:consensus-constant-technical}
For any $\epsilon>0$, the \textsc{CheapQuantumConsensus} algorithm 
solves Consensus against $n/3$-adversary in $O\Big(\big(\frac{1}{\epsilon}\big)^{4}\Big)$ rounds in expectation while using $O(n^{3\epsilon})$ qubits and communication bits per process (amortized), whp.

\end{theorem}
\begin{proof}
First, we argue for correctness. Compared to the protocol of Bar-Joseph and Ben-Or~\cite{Bar-JosephB02}, which works for an arbitrary number of failures $t\le n$, we changed the method of deriving random coin, c.f., line~\ref{line:quantum-flip}. Observe that even if \textsc{CheapQuantumCoin} fails to meet conditions of $t$-resilient coin flip, it always outputs some bit in every non-faulty processes. Thus, regardless of number of crashes the output could be an output of local coin flips with a non-zero probability. Since Bar-Josephs's and Ben-Or's algorithm works with probability $1$ (see Theorem~3 in~\cite{Bar-JosephB02}), thus \textsc{CheapQuantumConsensus} also achieves correctness with probability $1$.

Next, we estimate the expected number of phases (i.e. the number of iterations of the main while loop). We consider only \emph{good} phases, i.e. phases in which the adversary crashed at most $\frac{1}{10}$ fraction of processes that were correct at the beginning of this iteration. Note, that there can be at most $\frac{1}{3} / \frac{1}{10} < 4$ "bad" phases.
Let $x$ be the number of non-faulty processes at the beginning of some good phase. We consider following cases:
\\ \noindent Case \textit{a)} There exists a process that in this iteration executes line~\ref{line:undecided-one} or line~\ref{line:decided-one}. In this case, all other processes have to execute line~\ref{line:decided-one} or line~\ref{line:quantum-flip}, since the number of ones counted in different processes may differ by $\frac{x}{10}$ at most. All processes that in this iteration execute \textsc{CheapQuantumCoin} will set $b_{p}$ to $1$ with probability $\frac{1}{4}$ at least. What follows, in the next iteration all processes will start with $b_{p}$ set to $1$ and by Lemma~\ref{lem:same-b} the algorithms will decide within two next phases.
\\ \noindent Case \emph{b)} There exists a process that in this phase executes line~\ref{line:undecided-zero} or line~\ref{line:decided-zero}. Similarly to the previous case, we observe that all other processes have to execute line~\ref{line:decided-zero} or line~\ref{line:quantum-flip}, since, again, the number of ones counted in different processes may differ by $\frac{x}{10}$ at most. Therefore the same arguments applies, but know the final decision will be $0$.
\\ \noindent Case \emph{c)} None of processes executes one of lines~\ref{line:decided-one},~\ref{line:undecided-one},~\ref{line:undecided-zero}, or~\ref{line:decided-zero}. Thus, all processes participated in \textsc{CheapQuantumCoin} in line~\ref{line:quantum-flip}. By Theorem~\ref{thm:coin}, with probability at least $\frac{1}{4}$, all processes will set value $b_{p}$ to the same value. Thus, again by applying Lemma~\ref{lem:same-b}, we get that the algorithms will decide within two next phases.

We showed that if good phase happens, then the algorithm terminates within $2$ next iterations with probability at least $\frac{1}{4}$. Since there can be at most $4$ bad iterations, thus we can calculate the expected number of iterations as follows
\[
\E(ITE) = \sum_{i = 4}^{\infty} i \bigg(\frac{1}{4}\bigg)^{i} = O(1)
\ .
\]
Executing a single phase takes $O\Big(\big(\frac{1}{\epsilon}\big)^{4}\Big)$ rounds, which is the round complexity of the \textsc{FastCounting} algorithm and an upper bound of the time complexity of the \textsc{CheapQuantumCoin} algorithm, therefore the algorithm terminates in $O\Big(\big(\frac{1}{\epsilon}\big)^{4}\Big)$ rounds in expectation. Similarly, by taking the upper bounds on the communication complexity of the algorithms \textsc{FastCounting} and \textsc{CheapQuantumCoin} we get that the expected number of amortized communication bits used by the algorithm is $O(n^{3\epsilon})$.
\end{proof}
Finally, let us analyze the modified version of the algorithm \textsc{CheapQuantumConsensus} with the difference that processes use parameters $x := 2, d = \log{n}, \alpha := \log{n}$ in lines~\ref{line:fast-counting} and \ref{line:quantum-flip}.

\begin{theorem}
\label{thm:consensus-polylog-technical}
The modified version of the \textsc{CheapQuantumConsensus} algorithm 
solves Consensus against any adversary in $O(\log^{5}n)$ rounds in expectation while using $O(\log^{8}{n})$ qubits and communication bits per process (amortized), whp.

\end{theorem}
\begin{proof}
For the correcntess we argue exactly the same as in the proof of previous Theorem.
We also define good and bad phases, with only this differences that now the number of bad phases is at most $O(\log{n})$, since the adversary has the full power of crashing arbitrary number of processes. This being said we get that, by the very same reasoning, that the expected number of phases is
\[
\E(ITE) = \sum_{i = \log{n}}^{\infty} i \bigg(\frac{1}{4}\bigg)^{i} = \Theta(\log{n})
\ .
\]
By examining the time and bits complexity of the algorithms \textsc{FastCounting} and \textsc{CheapQuantumCoin} with parameters $x = 2, d = \log{n}, \alpha = \log{n}$, we get a single phase lasts $O(\log^{4}{n})$ rounds and contributes $O(n\log^{7}{n})$ bits to the total communication complexity. The latter, after dividing by $n$, gives the sought complexity amortized per process. Thus, the theorem follows.
\end{proof}
\section{Omitted proofs from Section~\ref{sec:quantum-coin}}\label{sec:coin-analysis}
\begin{proof}[Proof of Lemma~\ref{lem:surviving-probing}]
Consider two last iterations of the epoch $i$. Since the variable $\texttt{degree}_{p}$ has not changed it the epoch, thus also the variable $\texttt{adaptive\_degree}_{p}$ has remained unchanged in the epoch. In particular, examining line~\ref{line:while} assures, that in the last two testing rounds of this epoch, process $p$ received at least $\delta_{\alpha}$ responses from processes that had the variable $\texttt{adaptive\_degree}$ at least $\texttt{degree}_{p}$. 

Next, we proceed by induction. We claim, that in the $2j$ and $2j - 1$, for $1 \le j \le \gamma_{\alpha}$ last testing rounds there existed $(j, \delta_{\alpha})$-dense-neighborhood of $p$ with the properties the every processes belonging to this neighborhood is non-faulty and has its variable $\texttt{adaptive\_degree}$ at least $\texttt{adaptive\_degree}_{p}$. The induction step goes as follows. Assume that there exists $(j, \delta_{\alpha})$-dense-neighborhood of $p$ with the mentioned properties. In particular in the $2(j+1), 2(j+1) - 1$ last testing rounds, every process from this set had to receive at least $\delta_{\alpha}$ responses from processes with the variable $\texttt{adaptive\_degree}$ at least as big as the variable \texttt{adaptive\_degree} of the process. This follow from inspecting line~\ref{line:while}. The set of processes who responded with large values of \texttt{adaptive\_degree} constitute to the set $N^{j + 1}(p)$.

Eventually, we obtain the existence of the set $N^{\gamma_{\alpha}}(p)$ which satisfies the property of $(\gamma_{\alpha}, \delta_{\alpha})$-dense-neighborhood. Each process from the neighborhood has the variable $\texttt{degree}$ at least $\texttt{degree}_{p}$, since for every process $q$ it holds that $\texttt{degree}_{q} \ge \texttt{adaptive\_degree}_{q}$.
\end{proof}

\begin{proof}[Proof of Lemma~\ref{lemma:small-sets}]
Assume to the contrary that there exists an epoch such that more than $\frac{16n}{3\alpha^{i-2}}$ processes start this epoch with the variable $\texttt{degree}$ set to a value greater than $d\alpha^{i}$. Assume also w.l.o.g. that $i$ is the \emph{first} such epoch and let $A$ be the set of processes that have the variable $\texttt{degree}$ set to at least $d\alpha^{i}$ at the beginning of this epoch. Let $C$ be any survival set for $A$ with respect to the graph $\mathcal{G}(d\alpha^{i - 1})$. Note that since $A$ is a set of processes that are correct in epoch $i$, thus the processes from $C$ have been behaving correctly in epoch $1, \ldots, i -1 $ inclusively. Since the graph $\mathcal{G}(d\alpha^{i - 1})$ is $(16\frac{n}{d\alpha^{i - 2}},3/4,\frac{2}{3}\alpha)$-compact and $|A| \ge 16\frac{n}{d\alpha^{i - 2}}$, thus $C$ exists. At the beginning of the epoch $i$, the variable $\texttt{degree}$ of all processes from $C$ is greater than $d\alpha^{i}$, thus there must be a round $j < i$ in which the last process from $C$, say $q$, increases its variable $\texttt{degree}$ to $d\alpha^{i}$. But this gives us the contradiction. In the epoch $j$, all processes from $C$ operates in the graph $\mathcal{G}(d\alpha^{i - 1})$, thus they have at least $\delta_{\alpha}$ neighbors in $C$. All these neighbors have the variable $\texttt{degree}$ greater than $d\alpha^{i-1}$. In particular, in this epoch, the process $q$ will not execute line~\ref{line:if} and therefore it will not increase its variable $\texttt{degree}_{p}$ which give a contradiction with the maximality of $j$ and in consequence the minimality of $i$.
\end{proof}

Note that in the above proof of Lemma~\ref{lemma:small-sets} we use the fact that a suitable set $C$ of processes that have been correct throughout the whole epoch exists. We may choose this set and argue about it after the epoch, as the communication pattern in the algorithm is deterministic. Hence, in any round of the epoch,  processes in $C$ have at least as many non-faulty neighbors in communication graph as they have neighbors from set $C$. We use this number of neighbors in $C$ as a {\em lower bound} to argue about behavior of variables $\texttt{degree}$; therefore, our arguments do not depend on behavior of processes outside of $C$ and whether/when some of them fail during the epoch.

\begin{lemma}\label{lemma:non-faulty-ent}
Any two non-faulty processes $p$ and $q$ were connected by a quantum path of communication during the run of the algorithm. 
\end{lemma}
\begin{proof}
First observe, that the variable $\texttt{degree}_{p}$ can take at most $k + 1$ different values. Since there is $(k + 2)^{2}$ epochs, thus by the pigeonhole principle there must be a sequence of consecutive $k + 2$ epochs in which the variable $\texttt{degree}_{p}$ remains unchanged. Similarly, among these $k + 2$ epochs there must be two epochs in which the variable $\texttt{degree}_{q}$ remains unchanged. Let us denote these two epochs $i$ and $i + 1$. Since the variable $\texttt{degree}_{p}$ has not changed in the epoch $i + 1$, thus by applying Lemma~\ref{lem:surviving-probing}, we get that there exists a $(\gamma_{\alpha}, \delta_{\alpha})$-dense-neighborhood of $p$ in the graph $\mathcal{G}(\texttt{degree}_{p})$, say $A$, consisting of processes that a) were non-faulty in the epoch b) their variable $\texttt{degree}$ were at least $\texttt{degree}_{p}$. An immediate consequence is that all processes from $A$ were non-faulty in the epoch $i$. Moreover, since $A$ is $(\gamma_{\alpha}, \delta_{\alpha})$-dense-neighborhood of $p$ in the graph $\mathcal{G}(\texttt{degree}_{p})$ thus its size is at least $\frac{n\alpha}{\texttt{degree}_{p}}$ (Lemma~1~\cite{CK-PODC09}). Let $B$ be the analogues $(\gamma_{\alpha}, \delta_{\alpha})$-dense-neighborhood of $q$ derived from the properties of the graph $\mathcal{G}(\texttt{degree}_{q})$.

Now, consider quantum communication between sets $A$ and $B$ in the first testing round of the epoch $i + 1$. Processes from $A$ use the graph $\mathcal{G}(\texttt{degree}_{p})$ to communicate (or a denser graph) while processes from $B$ use the graph $\mathcal{G}(\texttt{degree}_{q})$ (or a denser graph). The graph $\mathcal{G}(\texttt{degree}_{p})$ is $(\frac{n\alpha}{\texttt{degree}_{p}})$-expanding and the graph $\mathcal{G}(\texttt{degree}_{q})$ is at least $(\frac{n\alpha}{\texttt{degree}_{q}})$-expanding. Therefore, at least one process from set $A$ inquires a process from set $B$ when executes line~\ref{line:inquires} in this testing round if $\texttt{degree}_{p} \ge \texttt{degree}_{q}$; or vice-versa if $\texttt{degree}_{p} < \texttt{degree}_{q}$. Since the set $B$ has diameter $\gamma_{\alpha}$ (Lemma~1~\cite{CK-PODC09}) and there remains $\gamma_{\alpha}$ testing rounds in the epoch $i +1$, thus by the end of this epoch any quantum messages send from $p$ could reach $q$.
\end{proof}

\begin{proof}[Proof of Theorem~\ref{thm:coin}]
Let $\mathcal{H} \subseteq \mathcal{P}$ be the set of initially non-faulty processes. Assume that at least $\frac{2}{3}|\mathcal{H}|$ of them remains non-faulty during the execution of the algorithm. By Lemma~\ref{lemma:non-faulty-ent}, any pair of processes from $\mathcal{H}$ is connected by quantum communication, therefore applying Lemma~\ref{lem:chance-win} to this set, gives that there is at least $\frac{2}{3} - o(1)$ (which is greater than $\frac{1}{2}$ for sufficiently large n) probability that all non-faulty processes return the same output bit. Since $0$ and $1$ are equally likely, thus the probabilistic guarantee on the weak global coin follows.
 
The number of rounds is deterministic and upper bounded by $O(k^2 \cdot \gamma_{\alpha}) = O\Big( \big(\frac{\log(n / d)}{\log{\alpha}}\big)^{2} \frac{\log{n}}{\log{\alpha}} \Big) = O\Big( \big(\frac{\log{n} }{\log{\alpha}}\big)^{3}\Big)$. To bound the communication complexity, assume that each graph $\mathcal{G}(d\alpha^{i})$, for $0 \le i \le k$, satisfies the desired expanding properties listed in the description of the algorithm. This, by the union bound argument, holds whp. By Lemma~\ref{lemma:small-sets} at the beginning of each epoch there is at most $\frac{16n}{d\alpha^{i - 2}}$ processes that inquires more than $d\alpha^{i}$ other processes in testing rounds of this epoch, for each $0 \le i \le k$. Since each message uses at most $O(\log{n})$ bits and qubits, thus a single testing round in an epoch adds at most
$$\sum_{i = 0}^{k} \frac{16n}{d\alpha^{i - 2}} \cdot d\alpha^{i} \le 16kd\alpha^{2} \cdot n\log{n}$$
qubits and bits to communication complexity. Since there is exactly $O\Big( \big(\frac{\log{n} }{\log{\alpha}}\big)^{3}\Big)$ testing rounds, thus the $O\Big( \big(\frac{\log{n} }{\log{\alpha}}\big)^{4}d\alpha^{2} \cdot n\log{n}\Big)$ upper bound on the total communication complexity of the algorithm follows. Dividing the latter by $n$ we receive amortized formula per process. %
\end{proof}

\section{The Analysis of FastCounting algorithm}\label{sec:fast-counting-analysis}
\begin{proof}[Proof of Theorem~\ref{thm:fuzzy-counting}]
Let us first argue for the correctness. If $|\mathcal{P}| = 1$, then the conditional statement in line~\ref{line:base-case} proves the correctness. Next, we proceed by induction. Assume that the correctness is proved for any subset $\mathcal{P}' \subseteq \mathcal{P}$. Therefore, whenever a process $p$ executes a recursive call in line~\ref{line:recursive-counting}, its variable $\texttt{active}_{p}$ is assigned to the number being an estimation of non-faulty active processes in the set $\mathcal{P}_{g}$. Consider now set $\{ \texttt{active}_{p} : p \in \mathcal{P} \}$. This is the set of input messages to the execution of the Gossip algorithm, see line~\ref{line:fast-gossip}. Thus, by Theorem~\ref{thm:fast-gossip}, we get that unless all process from a group $\mathcal{P}_{p}$ crashes, then all other non-faulty processes will have the number $\texttt{active}_{p}$ in their sets $\mathcal{R}$ (line~\ref{line:fast-gossip}). Therefore, the sum returned in line~\ref{line:counting-output} is the actual estimate of the number of all-non faulty processes in the set $\mathcal{P}$, since it can differ from the number of all processes starting the algorithm only by the number of crashed processes.

The bounds on time and communication complexity follows immediately from the bounds on complexity of the Gossip algorithm proposed in Theorem~\ref{thm:fast-gossip} and the fact that depth of the recursion at most $O(\log{n} / \log{x})$.
\end{proof}
\section{Further Results}
\label{sec:auxiliary}

\remove{
\jo{Reviewer 1: Are propositions 1 and 2 correct? I could not follow the argument. It is written that the adversary crashes any process to which p wants to send a message or who sends a message to p. But then why are you counting more than a single such message and saying that it is at most $o(n)$? If the process is crashed for the first message it wants to send to p then it definitely does not send more. I may be missing something here but it is really hard to understand the arguments.

Our response:
 The sum of probabilities is for fixed p over other processes q (sorry for not being precise), so by contradictory assumption it is o(n) with a constant>0 probability. Hence, there is a pair p,q with all four probabilities $o(1)$, which further leads to violating agreement with probability $1-o(1)$ in some execution (by indistinguishability of some executions by p,q). 
}
}

\begin{lemma}
\label{lem:worst-case-comm}
For any Consensus algorithm there is a strategy of an adaptive adversary under which some process sends $\Omega(n)$ messages with non-zero constant probability.  
\end{lemma}

\begin{proof}
Suppose otherwise.
Consider any process $p$ and two executions, in one of them $p$ starts with value $0$ and in the other starts with $1$. In both executions, the adversary fails all processes that want to send a message to $p$ and all to whom $p$ sends a message.
For each pair $(p,q)$, we obtain two probabilities $\alpha(p,q),\beta(p,q)$, of $p$ sending a message to $q$ in the first and the second execution.
The sum of probabilities $\sum_{q \in G} \alpha(p,q) + \beta(p,q)$ is $o(n)$ with constant probability $c > 0$, by contradictory assumption (recall that this sum is for the fixed process $p$).

Repeat the above for every process $p$. This way we obtain a complete weighted graph (by weights $\alpha(\cdot,\cdot)$ and $\beta(\cdot,\cdot)$). Again, by contradictory assumption, the total sum of all weights is $o(n^2)$ with the same constant probability $c > 0$, therefore, because the number of pairs is $n(n-1)$, there is a pair $p,q$ such that
\[
\alpha(p,q)+\alpha(q,p)+\beta(p,q)+\beta(q,p) = o(1)
\ .
\]
Consider an execution $\cE$ in which the adversary crashes all but processes $p,q$ in the very beginning. It assigns value $0$ to $p$ and value $1$ to $q$. It follows that in this execution $p,q$ try to communicate directly with probability $\alpha(p,q)+\beta(q,p) = o(1)$.
With complementary probability, $1-o(1)$, they do not communicate. 
In this case, $p$ cannot distinguish this execution $\cE$ from another execution $\cE_p$ in which all processes start with value $0$ but the adversary crashes all of them but $p$ in the beginning -- hence, by validity condition of consensus, $p$ must decide $0$ in both of these executions.
Similarly, $q$ cannot distinguish the constructed execution $\cE$ from execution $\cE_q$ in which all processes start with value $1$ but the adversary crashes all of them but $q$ in the beginning -- hence, by validity condition of consensus, $q$ must decide $1$ in both of these executions. Both the above facts must hold with probability $1$, as validity is required with probability $1$.
Consequently, in the constructed execution $cE$, $p$ decides $0$ while $q$ decides $1$, with probability $1-o(1)$, which violates agreement condition of consensus. This contradiction concludes the proof of the lemma. 
\end{proof}

Next, we argue that no constant-time Consensus algorithm uses amortized polylogarithmic number of communication bits, per process.

\begin{lemma}
\label{prop:constant-time}
Any quantum algorithm solving Consensus in expected time $O(1)$ requires a polynomial number of messages per process (amortized), whp.
\end{lemma}

\begin{proof}
Consider a Consensus algorithm working in expected $\tau$ number of rounds, for some $\tau=O(1)$.
By Markov inequality, the algorithm terminates by time $2\tau$ with probability at least $1/2$, for every strategy of the (adaptive) adversary.
We base on the idea in the proof of Lemma~\ref{lem:worst-case-comm}, but we do it for every round $r\le 2\tau$ and for number of sent messages in a round being $r\cdot n^{r/(2\tau)}$. Using recursive argument, we argue that there is a set $A_r$ of at least $\Omega(\frac{n}{r\cdot n^{r/(2\tau)}})$ processes (instead of just a pair of elements as in the proof of Lemma~\ref{lem:worst-case-comm}) that communicate with each other with probability $o(1)$, with already sending $\Omega(n^{(r-1)/(2\tau)})$ messages, amortized per process in this set.
In one of the rounds, however, a constant ($1/(2\tau)$) fraction of alive processes must increase their communication peers by factor of at least $n^{1/\tau}$, because in $\tau$ expected rounds (so also with constant probability) they need to contact $\Omega(n)$ other peers in some executions -- 
again, by the same argument as in the proof of Lemma~\ref{lem:worst-case-comm} (i.e., otherwise, we could find a pair of them, assign different values, and enforce decision on them with a constant probability, violating agreement condition). 
In such round, the adversary allows them to send their message without crashes, thus enforcing total number of 
$\Omega(\frac{n}{r\cdot n^{r/(2\tau)}} \cdot n^{(r-1)/(2\tau)} \cdot n^{1/\tau})=\Omega(n^{1+1/(2\tau)})$ messages, i.e., at least $\Omega(n^{1/(2\tau)})$ messages amortized per every process in the system. As this happens with a constant probability, the proof follows.
%
%
\end{proof}

\end{document}